\documentclass[runningheads]{llncs}

\usepackage{graphicx}

\usepackage{algorithm,algpseudocode}
\usepackage{algorithm}
\usepackage{algpseudocode}
\usepackage[linesnumbered,ruled,vlined,algo2e]{algorithm2e}

\usepackage{booktabs}
\usepackage{tabularx}
\usepackage{gensymb}
\usepackage{graphicx}
\usepackage{enumerate}
\usepackage{hyperref}
\usepackage{multirow}
\usepackage{xcolor}
\usepackage[labelformat=simple]{subcaption}

\usepackage{mathrsfs}
\usepackage{comment}

\usepackage{multicol}
\usepackage{graphicx,tipa}
\usepackage{amsfonts}
\usepackage{amsmath}
\usepackage{tcolorbox}

\begin{document}

\title{Treasure Hunt in Graph using Pebbles}
 \author{Adri Bhattacharya\inst{1} \and
 Barun Gorain\inst{2} \and
 Partha Sarathi Mandal\inst{1}\orcidID{0000-0002-8632-5767}}
 \authorrunning{Bhattacharya et al.}
 \institute{Indian Institute of Technology Guwahati, 781039, India\\ \email{\{a.bhattacharya, psm\}@iitg.ac.in}
 \and
 Indian Institute of Technology Bhilai, India\\
 \email{barun@iitbhilai.ac.in}
 }
\maketitle              
\begin{abstract}
In this paper, we study the treasure hunt problem in a graph by a mobile agent. The nodes in the graph $G=(V,E)$ are anonymous and the edges incident to a vertex $v\in V$ whose degree is $deg(v)$ are labeled arbitrarily as $0,1,\ldots, deg(v)-1$. At a node $t$ in $G$ a stationary object, called {\it treasure} is located. The mobile agent that is initially located at a node $s$ in $G$, the starting point of the agent, must find the treasure by reaching the node $t$. The distance from $s$ to $t$ is $D$. The {\it time} required to find the treasure is the total number of edges the agent visits before it finds the treasure. The agent does not have any prior knowledge about the graph or the position of the treasure. An oracle, that knows the graph, the initial position of the agent, and the position of the treasure, places some pebbles on the nodes, at most one per node, of the graph to guide the agent towards the treasure.

This paper aims to study the trade-off between the number of pebbles provided and the time required to find the treasure. To be specific, we aim to answer the following question:

\begin{itemize}
    \item ``What is the minimum time for treasure hunt in a graph with maximum degree $\Delta$ and diameter $D$ if $k$ pebbles are placed? "
\end{itemize}

We answer the above question when $k<D$ or $k=cD$ for some positive integer $c$. We design efficient algorithms for the agent for different values of $k$. We also propose an almost matching lower bound result for  $k<D$.

\textbf{Keywords}: Treasure Hunt,
Mobile Agent,
Anonymous Graph,
Pebbles,
Deterministic Algorithms 
\end{abstract}

\section{Introduction}

\subsection{Background and Motivation}\label{motivation}

Treasure hunt problem is well studied in varying underlying topologies such as graphs and planes \cite{bouchard2020almost,bouchard2020deterministic,emek2015many,gorain2021pebble,miller2015tradeoffs,pelc2018reaching,pelc2018information,pelc2019cost}.
In this paper, we have delved into the treasure hunt problem using mobile agents on graphs. The main idea of this problem is that the mobile agent starting from a position has to find a stationary object, called treasure, placed at some unknown location in the underlying topology. There are many real-life applications to this problem. Consider a scenario, where a miner is stuck inside a cave and needs immediate assistance. In network applications, consider a network containing a virus, the agent in this case is a software agent whose task is to find the virus located in some unknown location inside the network.
For any graph with maximum degree $\Delta$. The agent can find the treasure, located $D$ distance away by performing a simple breadth-first search (BFS) technique in $\mathcal{O}(\Delta^D)$ time. But this is naive strategy is expensive as many real-life problems require a much more efficient solution. Suppose a person is stuck inside a building that has caught fire. He needs to find the fire exit and then evacuate within a short period of time. These kinds of emergencies require a faster solution. The person needs external help guiding him toward the fire exit. Similar to that finding the treasure, the agent needs some external help to guide the agent toward the treasure. That help is provided to the mobile agent by the oracle. The external information provided by the oracle is in the form of pebbles placed at the graph's vertices (nodes) \cite{gorain2021pebble}, also termed as {\it advice}. This advice guides the agent towards the treasure. The pebbles are placed at the nodes, so the agent visiting the nodes gain some knowledge and find the treasure using that information. The oracle places pebbles, at most one at a node, by knowing the underlying graph topology, initial position of the agent, and treasure's location. 
Recently studied in the paper by Gorain et al. \cite{gorain2021pebble} the treasure hunt problem in anonymous graph. They studied the question, what is the fastest treasure hunt algorithm regardless of any number of pebbles placed. In that paper, they obtained a faster algorithm that finds the treasure, regardless of the number of pebbles placed. So, now a natural question arises which they did not address. Given $k$ many pebbles, what is the fastest possible treasure hunt algorithm.

In this paper, we find the solution of the question:
\textit{Given $k$ pebbles what is the fastest algorithm which solves treasure hunt problem in anonymous graph. }

\subsection{Model and Problem Definition}\label{model}

Search domain by the agent for finding the treasure is considered as a simple undirected connected graph $G=(V,E)$ having $n=|V|$ vertices which are anonymous, i.e., they are unlabeled. The vertices are also termed as nodes in this paper. An edge $e=(u,v)$ must have two port numbers one adjacent to $u$, which is termed as {\it outgoing} port from $u$  and the other adjacent to $v$, termed as {\it incoming} port of $v$ (refer the edge $(v_i,v_{i+1})$ in Fig. \ref{Fig-Tree2Dpebble}, where $\rho_4$ is the outgoing port from $v_i$ and $\rho_0$ is the incoming port of $v_{i+1}$). $\Delta$ is denoted as the maximum degree of the graph. Initially, the agent only knows the degree of the initial node. A node $u\in V$ with $deg(u)$ is connected with its neighbors, $u_0,u_1,\cdots, u_{\{deg(u)-1\}}$ via port numbers which have arbitrary but fixed labelings $\rho_0, \rho_1,\cdots,\rho_{deg(u)-1}$, respectively. Agent visiting a node can read the port numbers when entering and leaving a node, as stated in the paper \cite{dereniowski2012drawing}. Moreover, when the agent reaches a node $v$ from a node $u$, it learns the outgoing port from $u$ and the incoming port at $v$, through which it reaches $v$.
The first half of $u$'s neighbors are the nodes corresponding to the outgoing port numbers $\rho_0,\rho_1,\cdots,\rho_{\frac{deg(u)}{2}-1}$, whereas the second half of $u$'s neighbors are the nodes corresponding to the outgoing port numbers $\rho_{\frac{deg(u)}{2}},\cdots,\rho_{deg(u)-1}$. The agent is initially placed at a node $s$. The treasure $t$ is located on a node of $G$ at a distance $D$ from $s$, which is unknown to the agent. The oracle places the pebbles at the nodes of the graph $G$, in order to guide the agent towards the treasure. At most one pebble is placed at a node. The agent has no prior knowledge about the underlying topology nor it has any knowledge about the position of treasure, pebble positions as well as the number of pebbles deployed by the oracle. The agent has no knowledge about the value of $D$ as well. The agent can only find the treasure or pebble whenever it reaches the node containing the treasure or pebble. Distance is considered as the number of edge traversal. We denote the shortest distance between any two nodes by {\it dist}, i.e., between nodes $u,v\in G$ as $dist(u,v)$, hence $dist(s,t)=D$.  Moreover the agent has unbounded memory. The {\it time} of treasure hunt is defined as the number of edge traversed by the agent from its initial position until it finds the treasure.

\subsection{Contribution}\label{contribution}
 We study the trade-off between the number of pebbles provided by the oracle and the associated time required to find the treasure. The contributions in this paper is mentioned below.

\begin{itemize}
\item For $k<\frac{D}{2}$ pebbles, we propose an algorithm that finds the treasure in a graph at time $\mathcal{O}(D \Delta^{\frac{D}{(2\eta+1)}})$, where $\eta=\frac{k}{3}$.
\item For $\frac{D}{2} \le k < D$, we propose a treasure hunt algorithm with time complexity $O(k \Delta^{\frac{D}{k+1}})$.
\item In case of bipartite graphs, the proposed algorithm for treasure hunt has time complexity $\mathcal{O}(k\Delta^{\frac{D}{k}})$ for $0 \le k <D$.
\item For $k=c D$ where $c$ is any positive integer, we give an algorithm that finds the treasure in time 
    $\mathcal{O}\left[cD{(\frac{\Delta}{2^{{c}/{2}}})}^2 + cD\right]$ 

\item  We propose a lower bound result  $\Omega((\frac{k}{e})^{\frac{k}{k+1}}(\Delta-1)^{\frac{D}{k+1}})$ on time of treasure hunt  for $0 \le k<D$.
    
\end{itemize} 

\subsection{Related Work}

Several works have been done on searching for a target by one or many mobile agents under varied underlying environments. The underlying environment can be a graph or a plane, also the search algorithm can be deterministic or randomized. The paradigm of \textit{algorithm with advice} was mainly studied for networks, where this advice (or information) enhances the efficiency of the solutions in \cite{emek2011online,fraigniaud2008tree,fraigniaud2010communication}. In the past few decades, the problem of treasure hunt has been explored in many papers, some of them are \cite{bouchard2020deterministic,miller2015tradeoffs,pelc2018information}. The book by Alpern et al. \cite{alpern2006theory} provides a brief survey about searching and rendezvous problem, for an inert target, where the target and the agent are both dynamic in nature and they cooperate among themselves to meet. This book mostly deals with randomized search algorithms. The treasure hunt problem is mainly studied in a continuous and discrete model. Bouchard et al. \cite{bouchard2020deterministic} studied the problem of treasure hunt in the Euclidean plane, where they showed an optimal bound of $\mathcal{O}(D)$ with angular hints at most $\pi$. Pelc et al. \cite{pelc2019cost} provided a trade-off between time and information of solving the treasure hunt problem in the plane. They showed optimal and almost optimal results for different ranges of vision radius. Pelc et al. \cite{pelc2018information} gave an insight into the amount of information required to solve the treasure hunt in geometric terrain at $\mathcal{O}(L)$- time, where $L$ is the shortest path of the treasure from the initial point. Further Pelc \cite{pelc2018reaching} investigated the treasure hunt problem in a plane with no advice for both static and dynamic cases. 
Spieser et al. \cite{6426279} studied the problem, where multiple self-interested mobile agents compete among themselves to capture a target distributed in a ring, they have minimal sensing capability and limited knowledge of where the target is positioned. In \cite{miller2015tradeoffs}, studies the amount of information available to the agent at priori and the time of rendezvous and treasure hunt problems in graph is studied.  Moreover, Dieudonné et al. \cite{dieudonne2015deterministic} explored the problem of rendezvous in a plane, where the adversary controls the speed of the agent. Georgiou et al. \cite{georgiou2016search} studied the treasure evacuation problem in unit disk with a single robot. Gorain et al. \cite{gorain2021pebble} studied the treasure hunt problem in the graphs with pebbles and also provided a lower bound of the run time complexity using any number of pebbles. Our problem is a more generalized version of the paper by Gorain et al. \cite{gorain2021pebble}, where they have used an infinite number of pebbles to give an almost optimal algorithm with time $\mathcal{O}(D\log\Delta+\log^3\Delta)$. 
This paper tries to find an efficient algorithm for a given number of pebbles.

The rest of the paper is organized as follows. In section \ref{k<D}, given $k<D$ pebbles, we provide treasure hunt algorithm and its analysis for a general graph. Further, in section \ref{k>D}, given $k\ge D$ we propose the treasure hunt algorithm for a general graph. In section \ref{LowerBound}, we have given the lower bound for the case $k<D$. Finally, concluded in section \ref{conclude}.

In the following sections, we propose different algorithms for different graph topology and their analysis.

\section{Treasure Hunt Algorithm when $k<D$}\label{k<D}

In this section, we provide algorithms and their analysis for the case when the number of pebble $k$ is less than $D$.
We first give an idea for tree topology and then extend our idea for the bipartite graph. Then further modify our idea for the general graph, with the introduction of a new paradigm termed as \textit{markers}.

\subsubsection{Idea of the pebble placement and algorithm in a tree network:}\label{tree}
 Let $T$ be a rooted tree with root $s$ and a node that is $i$ distance from $s$ is located at the  level $L_i$ of $T$. As the treasure $t$ is located $D$ distance away from $s$, hence $t$ 
is residing at some node 
 at level $L_D$. Consider the shortest path ($P$) from $s$ to $t$ is $P=v_0,v_1,\cdots,v_{D}$.

 \noindent \textbf{Without Pebble:} Consider the scenario, when there are no pebble placed in the underlying tree topology by the oracle. To find the treasure starting from $s$, the agent obeys a simple strategy. It naively searches all its neighbors and further down the level, i.e., its children, grandchildren and so on until the treasure is encountered. Now in the worst case, each node may have degree $\Delta$, hence the time of finding the treasure can be as worst as $\mathcal{O}(\Delta^D)$. To overcome this high time complexity, let us analyze how the placement of pebbles by the oracle helps in efficiently reaching the treasure. 
 
 \noindent \textbf{With Pebble:} Suppose the oracle places only a single pebble at a node in the first level $L_1$. In this case, after searching at most $\Delta$ neighbors of $s$, the agent will find a pebble at a node $v_1$ (say) at the level $L_1$, from which it further searches only the subtree rooted at $v_1$. This reduces the time to $\mathcal{O}(\Delta^{D-1})$. Now, if the pebble is placed at a node in the level  $L_i$ then the time of finding the treasure is $\mathcal{O}(\Delta^{D-i})$. So, it is simple to observe that this time is minimum when $i=\frac{D}{2}$. In this case, the treasure can be found at time $\mathcal{O}(\Delta^{\frac{D}{2}})$ which is a significant improvement from $\mathcal{O}(\Delta^D)$ after introducing a single pebble in an appropriate level. Further with the introduction of two pebbles and placing them at $\frac{D}{3}$ distance apart from each other, i.e.,  at a node in the levels $L_{\frac{D}{3}}$ and $L_{\frac{2D}{3}}$, respectively. The time taken to find the treasure further reduces to $\mathcal{O}(\Delta^{\frac{D}{3}})$. So, if the oracle places $k$ pebbles, $\frac{D}{k+1}$ distance apart along the path $P$. The treasure can be found in time $\mathcal{O}(\Delta^{\frac{D}{k+1}})$.  

But the aforementioned naive idea cannot be directly applied to general graphs. As the nodes in the graph are anonymous, i.e., there is no id for the nodes. The agent can't distinguish between a node that is visited or not. So, there is an issue to deal with: 
\begin{itemize}
    \item Suppose the agent is currently searching from some node containing a pebble at level $L_i$, then how to determine the fact that the pebble found is at level $L_j$, where $i<j$ but not $j<i$ or $j=i$.
\end{itemize}

Hence if these two issues are not resolved then in the worst case the agent may move inside a cycle for infinite time. So, in the next two sections, we deal with the issues related to the general graph. We provide algorithms and their analysis for the agent to find treasure when $\frac{D}{2}\le k<D$ and $k<\frac{D}{2}$. 
 
\subsection{$\frac{D}{2}\leq k < D$}\label{GeneralGraph-lessD}

In this case, if the oracle places a pebble along the path $P$ at alternative levels, i.e., at the nodes $v_j$, where $1\leq j \leq D$ and $j$ is even. The agent searches every possible path of length $\frac{D}{k+1}$ (=$l$) until a pebble or the treasure is encountered from {\it SearchNode}. The length of the path between two pebbles, in this case, is at most $2$, as $k\geq \frac{D}{2}$. So, by searching a path of length at most 2 from {\it SearchNode} (i.e., $l\le 2$), the agent cannot return to itself. The reason is, the graph $G$ has no multiple edges and self loops. Also, it cannot go to the previous {\it SearchNode}. It is because suppose the {\it SearchNode} is at level $L_i$, then all the incoming ports from the level $L_{i-1}$ to the {\it SearchNode} is already saved. Further, the agent cannot use these saved ports while searching a BFS of length $l$ from {\it SearchNode}. Hence the length of the path from $L_i$ to any node in $L_k$ (where $k<i$) containing pebble is at least $l+1$. So, both the issues of circling its way back to itself and going back are restricted. Hence, the agent can only move forward along $P$ and ultimately finds the treasure. The time taken to search all possible paths of length $\frac{D}{k+1}$ is $\mathcal{O}(\Delta^\frac{D}{k+1})$ and the searching is done from each of the $k$ pebbles. Hence the total time to find the treasure is at most $\mathcal{O}(k\Delta^\frac{D}{k+1})$.

Now in the case of a general graph with $\frac{D}{2}\leq k < D$ pebbles, the placement of pebbles at alternate levels ensures that there is no returning back and as well as going back to the {\it SearchNode}. But this fact is not valid for general graphs, when $k<\frac{D}{2}$ pebbles are provided. This is explained with the help of the following example. 

\noindent\textbf{Example:}\label{example} 
Consider the example in Fig. \ref{Fig-Impossible}, where the {\it SearchNode} is $v_3$ and $l=3$. The correct path from $v_3$ to the treasure is along $v_3\longrightarrow v_4 \longrightarrow v_5 \cdots v_j \cdots \longrightarrow t$. But when it performs a BFS of length $l$ from {\it SearchNode}. It will not be possible for the agent to distinguish between the paths $v_3\longrightarrow v_4 \longrightarrow v_5 \longrightarrow v_6$ and $v_3\longrightarrow u_4 \longrightarrow u_5 \longrightarrow v_3$. In both the cases after traversing a {\it dist} of $l$ from $v_3$, the agent encounters a pebble. In the worst case, the agent may traverse this wrong path each time and never reach the treasure. Moreover, the number of pebbles must be at least 2, as with a single pebble it is not possible for the agent to find the treasure. Consider the Fig. \ref{Fig-Impossible}, where a pebble is placed at only $v_3$. In this case, the agent may never find treasure. It is because, at every search from $v_3$, the agent may circle its way back to $v_3$ (as the nodes and pebbles are anonymous) rather than encountering the treasure. 
    
\begin{figure}[h]
\begin{minipage}[9cm]{.45\textwidth}

\centering
\includegraphics[width=0.6\textwidth]{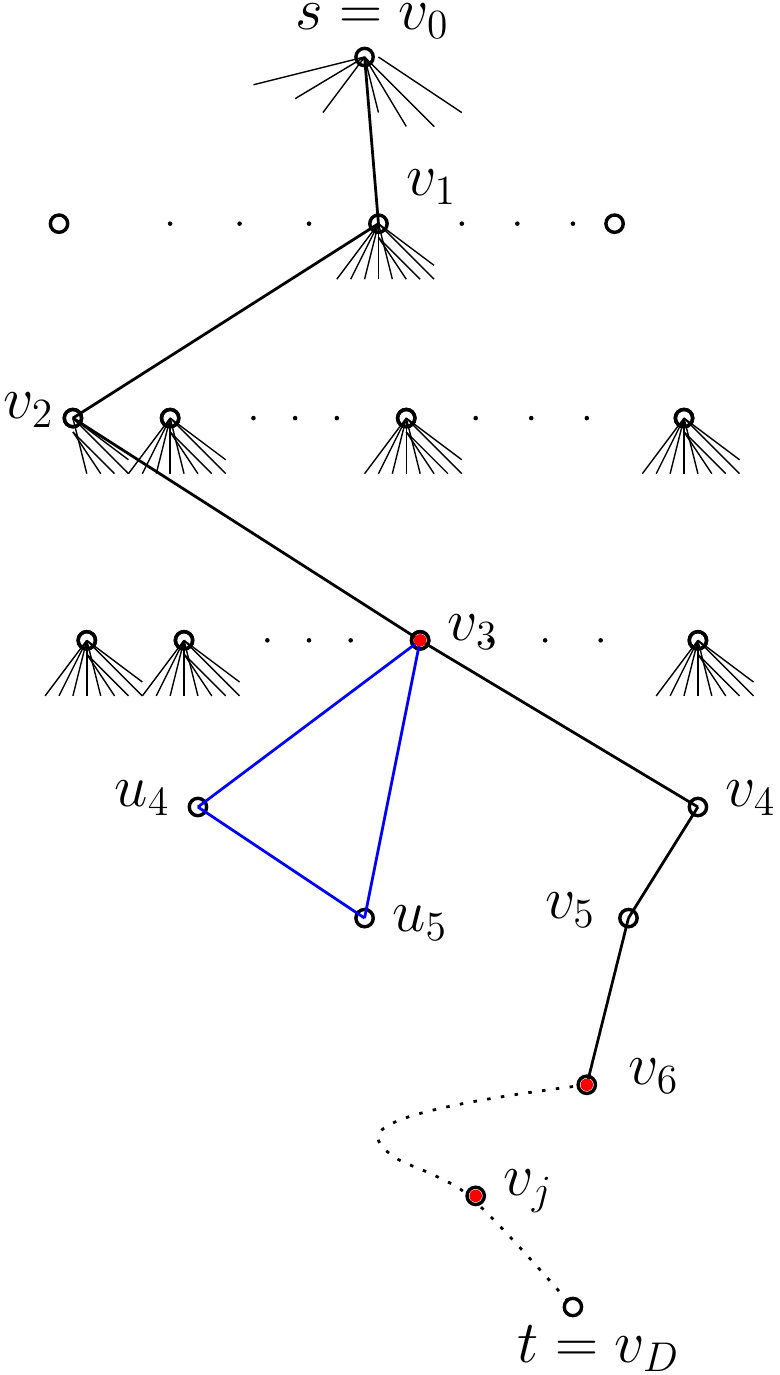}
\caption{Impossibility case in General Graph}
\label{Fig-Impossible}
\end{minipage}
\hspace{0.5cm}
\begin{minipage}[9cm]{.45\textwidth}
\centering
\includegraphics[width=1.2\textwidth]{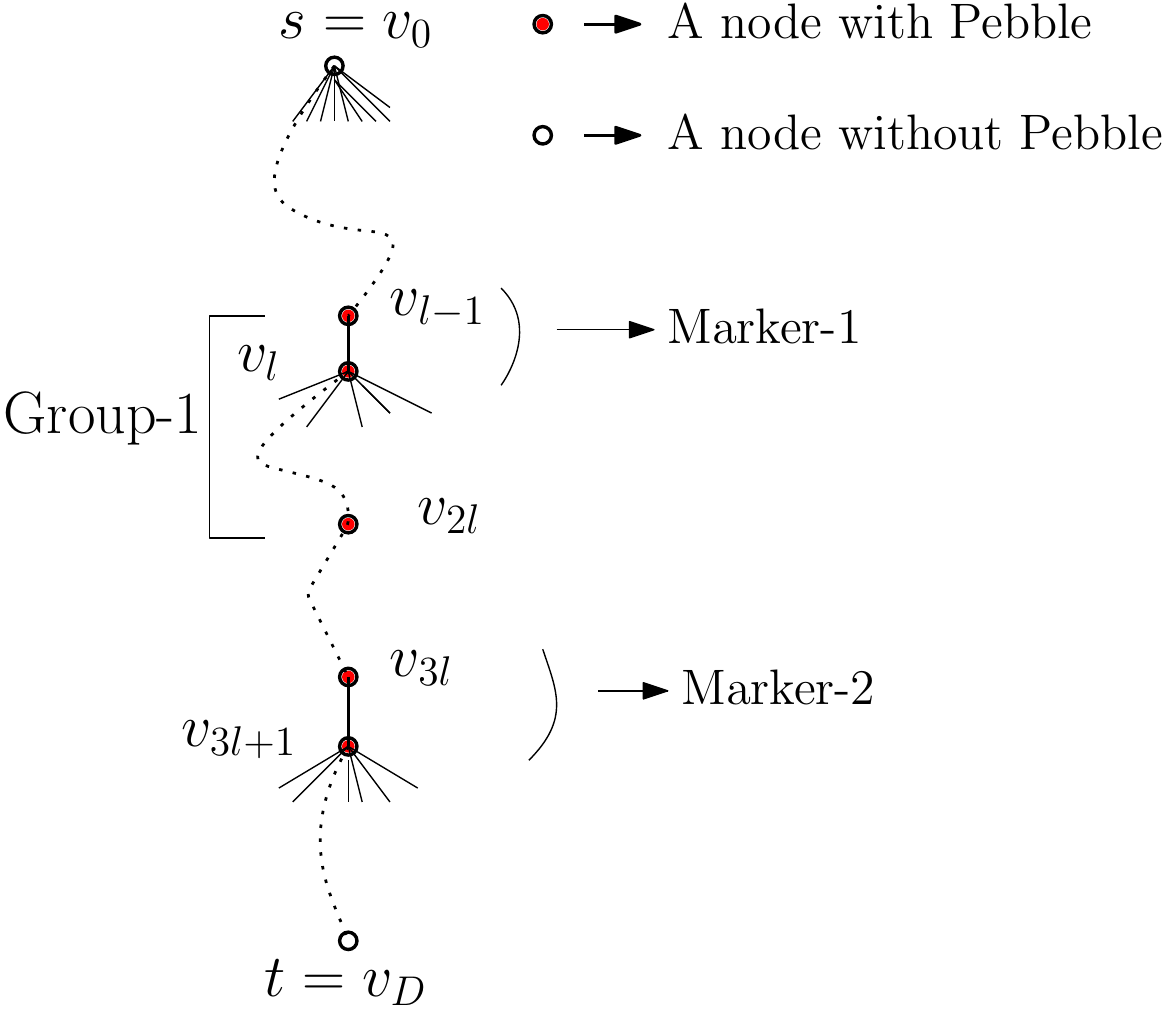}
\caption{Pebble Placement in General Graph with Multiple Marker}
\label{Fig-GeneralGraphwithMultiMarker}
\end{minipage} 

\end{figure}

\subsubsection{Marker:} To deal with this scenario we create a notion of markers. The idea of {\it markers} came from the concept of colors. Colors are always helpful in distinguishing certain characteristics. In our case it helps to identify whether a node is already visited. As there is no concept of using colors in our model. We replicate the idea of colors by using certain combination of pebbles, this combination is termed as {\it Marker}. To denote a marker, the oracle places two pebbles adjacent to each other along the path $P$. When the agent finds a pebble placed in one of the adjacent nodes of the current pebble, it understands that it has found a marker. We generalize our idea of finding treasure in general graphs with markers.

\subsection{$k<\frac{D}{2}$}\label{GeneralGraph-lesshalfD}

 As discussed above, the introduction of a marker helps the agent distinguish between visited and non-visited pebbles. In this section, we discuss how multiple marker helps the agent to find the treasure in a general graph.
The pebble placement strategy is discussed as follows. 

We define a group as a marker together with the immediate next pebble. Markers and pebbles are placed alternatively as shown in Fig. \ref{Fig-GeneralGraphwithMultiMarker}. Let $l$ be the distance between a marker and a pebble in a group. Further two such groups are also placed at $l$ distance apart. We denote $\eta$ to be the number of such groups. We can differentiate the following cases:
\begin{itemize}
    \item \textit{Case-1:} ($k=3\eta$) In this case, the oracle places the markers and the pebbles $l$ distance apart along the path $P$ (refer Fig. \ref{Fig-GeneralGraphwithMultiMarker}), where $l=\frac{D-\eta}{2\eta+1}$.
    \item \textit{Case-2:} ($k=3\eta+1$) Similarly, in this case $l=\frac{D-\eta}{2\eta+2}$.
    \item \textit{Case-3:} ($k=3\eta+2$) Similarly, in this case as well, $l=\frac{D-\eta-1}{2\eta+2}$.
\end{itemize}




Below is a detailed description of the algorithm \textsc{TreasureHuntForGraphWithMarker} that the agent executes to find the treasure.
\begin{enumerate}
    
   \item \label{step-1Marker} The agent starting from $s$, sets {\it SearchNode}=$s$ and performs a BFS in increasing lexicographic order of outgoing port numbers until a treasure or pebble is found.
    \item \label{step-2Marker}If the treasure is found, the algorithm terminates.
    \item \label{step-3Marker}If the treasure is not found and a pebble is found. Then the agent sets the {\it dist} between $s$ and this node containing the pebble as $l-1$. Also consider that the node containing the pebble is $v_{l-1}$ at $L_{l-1}$-th level.
    
    \item \label{step-4Marker}Further from $v_{l-1}$, the agent performs two tasks. 
    \begin{enumerate}
        \item \label{step-4aMarker}Firstly it searches the neighbors of $v_{l-1}$, and finds another pebble at the node $v_l$ in level $L_{l}$. Further, it stores the length of the path $l$ from {\it SearchNode} to $v_l$ and identifies a marker is found.
        \item \label{step-4bMarker}Secondly it stores the incoming port number $\rho_{l-1}$ of the incoming edge $(v_{l-1},v_l)$.
        
    \end{enumerate}
    \item \label{step-5Marker}Reset {\it SearchNode}=$v_l$.
     \item \label{step-6Marker}The agent performs a BFS of length $l$ from {\it SearchNode} until the treasure or pebble is found.
     \item \label{step-7Marker}If the treasure is found, then the algorithm terminates.
     \item \label{step-8Marker}If the treasure is not found and whenever a pebble is found, there are two possibilities:
     \begin{enumerate}
         \item \label{step-8aMarker}\textit{Possibility-1:} The agent has returned back to {\it SearchNode} as the underlying graph topology is a general graph.
         \item \label{step-8bMarker}\textit{Possibility-2:} The agent has encountered a new pebble along path $P$.
     \end{enumerate}
     Now to understand which of these possibilities the agent has encountered. The agent travels the stored sequence of a port number, in this case $\rho_{l}$ of length $1$. If in this traversal marker is found, then the agent has encountered {\it possibility-1} and searches a different path, otherwise if the marker is not found then it is {\it possibility-2}.
     \item \label{step-9Marker}If the agent encounters {\it possibility-2}, i.e., it has found a pebble for the first time at the node $v_{2l}$ in $L_{2l}$-th level. The agent performs the following tasks.
     \begin{enumerate}
         
         \item \label{step-9aMarker}It stores the sequence of incoming port numbers of the shortest path from $v_{l-1}\longrightarrow {\it SearchNode} \longrightarrow v_{2l}$ of length $l+1$.
         \item \label{step-9bMarker}Completes the BFS search at the $L_{{2l}-1}$ level. Whenever there is another pebble encounter. The agent stores the incoming port number of the incoming edge of that node containing a pebble, i.e., the edge with the node in level $L_{{2l}-1}$ and $v_{2l}$. Sets {\it SearchNode}=$v_{2l-1}$ 
     \end{enumerate}
     \item \label{step-10Marker}From {\it SearchNode} it performs a BFS of length $l$ using the ports except for the stored incoming ports until treasure or pebble is encountered.
     \item \label{step-11Marker}If treasure encountered go to step \ref{step-7Marker}.
     \item \label{step-12Marker}If pebble encountered then perform only step \ref{step-8Marker}. Further if {\it Possibility - 2} arises. Then search the neighbor of current {\it SearchNode} (i.e., $v_{3l}$).
     \item \label{step-13Marker}If a pebble is found at one of its neighbor node $v_{3l+1}$,  then identify a new marker is found and store the incoming port $\rho_{3l}$ of the incoming edge $(v_{3l},v_{3l+1})$. Then go to step \ref{step-5Marker}.
     \item \label{step-14Marker}Otherwise if no pebble is found then go to step \ref{step-9Marker}.
     
\end{enumerate}

\begin{lemma} \label{Correctness-GeneralGraphmultiMarker}
Given $ k<\frac{D}{2}$ pebbles, the agent following \textsc{TreasureHuntForGraphWithMarker} algorithm successfully finds the treasure in a general graph with the help of multiple marker.
\end{lemma}
\begin{proof}
To prove the correctness of our algorithm, we first ensure that while searching from {\it SearchNode} the agent always finds a new pebble, i.e., it does not encounter an already visited pebble. Secondly, we ensure that the agent will always find the same pebble at least once, by searching all possible paths of length $l$ from {\it SearchNode}. Finally, we ensure that the treasure is found, when the agent searches all possible paths of length $l$ from the node containing the last pebble placed along the path $P$. They are resolved in the following manner. 
\begin{itemize}
\item {\it Circle back to current {\it SearchNode}:} If the current {\it SearchNode}=$v_k$ is part of a marker, then it stores the incoming port number $\rho_k$ of the edge $(v_k,v_{k+1})$ (refer step \ref{step-4bMarker} of \textsc{TreasureHuntForGraphWithMarker} algorithm). Further, it performs a BFS from $v_{k+1}$, whenever a pebble is encountered, it checks the node with port $\rho_k$ (refer step \ref{step-8Marker} of \textsc{TreasureHuntForGraphWithMarker} algorithm) from the newly encountered pebble. If another pebble is found, i.e., a marker is found, then it concludes that the agent has circle back to current {\it SearchNode}, i.e., $v_k$. In this case, the agent cancels this path and continues its search along different path (refer {\it Possibility-1} in step \ref{step-8aMarker} of \textsc{TreasureHuntForGraphWithMarker} algorithm).

Otherwise, if the current {\it SearchNode}=$v_k$ is not part of marker, then the agent has a stored sequence of incoming port numbers of length $jl+1$ (refer step \ref{step-9Marker} of \textsc{TreasureHuntForGraphWithMarker} algorithm), where $j$ is the number of pebbles along $P$ from the last marker to the {\it SearchNode}. So, the agent from {\it SearchNode} performs a BFS of length $l$. Whenever a pebble is encountered. The agent traverses the already stored sequence of port numbers to check whether it has reached the marker. If the marker is reached, it concludes that the agent has circled its way back to {\it SearchNode} and searches a different path. Otherwise, the agent has moved to the new {\it SearchNode}. Hence the use of markers prevents the agent to circle its way back to itself.

\item {\it Circle back to previous {\it SearchNode}:} Let at some iteration of the algorithm, current {\it SearchNode} is $v_i$ at level $L_i$ of the BFS tree. Then the pebble encountered while searching all possible paths of length $l$ should be at level $L_j$, where $j>i$. This is because, the length of the path from the last {\it SearchNode} to $v_i$ is at least $l+1$. The reason being, the agent does not travel through the edges adjacent to $v_i$ whose ports are saved (refer step \ref{step-10Marker} of \textsc{TreasureHuntForGraphWithMarker} algorithm). This implies that the new search node is always at $L_j$-th level, where $j>i$. So, this means that the agent cannot go back to already visited pebbles. 

\item {\it Guaranteed finding of pebble at length $l$ from current {\it SearchNode}:} The oracle places the pebbles and markers $l$ distance apart. So, the agent while searching from {\it SearchNode}, all possible paths of length $l$, must encounter one pebble.
    
\item {\it Guaranteed finding of treasure at length $l$ from $k$-th pebble:} The node containing the last pebble is chosen in such a manner by the oracle, such that the {\it dist} between treasure and that node is at most $l$. So, the agent while searching from {\it SearchNode} all paths of length $l$ must find the treasure.
\end{itemize}
All points above guarantee that the agent successfully finds the treasure.
\qed
\end{proof}

\begin{theorem}
The agent finds the treasure in $\mathcal{O}(D \Delta^{\frac{D}{(2\eta+1)}})$ time, where $\eta=\frac{k}{3}$. 
\end{theorem}

\begin{proof}
Starting from $s$, the agent performs a BFS and encounters a pebble at distance $l$ for the first time. Further, whenever the agent encounters either a marker (combination of two pebbles) or a single pebble, it searches all possible paths of length $l$. The time taken to search all paths of length $l$ from a node is $\mathcal{O}(\Delta^l)$, where $\Delta$ is the maximum degree of $G$. Whenever a single pebble is encountered, the agent traverses the sequence of port numbers of length $l+1$ (refer step \ref{step-8Marker} of \textsc{TreasureHuntForGraphWithMarker} algorithm). In the worst case, the agent traverses a path of length $l+1$ for each $\mathcal{O}(\Delta^l)$ many searches. So, the time taken is $\mathcal{O}((l+1) \Delta^l)$. The agent has to search more, when the number of single pebble is $\eta+1$, i.e., $l=\frac{D-\eta}{2\eta+2}$ (refer pebble placement strategy). By lemma \ref{Correctness-GeneralGraphmultiMarker}, the agent successfully reaches the treasure. Hence the total time taken to find the treasure is $\mathcal{O}((\eta+1)(l+1) \Delta^l)=\mathcal{O}\left(D\Delta^{\frac{D-\eta}{2\eta+2}}\right)=\mathcal{O}\left(D\Delta^{\frac{D}{2\eta+1}}\right)$, where $\eta=\frac{k}{3}$.
\qed
\end{proof}

In the following part, we provide a close to optimal treasure hunt algorithm for a special class of graphs. We show that in bipartite graphs, treasure can be found without the placement of markers, unlike in a general graph, where the treasure hunt is not possible without markers. Let $G$ be the bipartite graph. Consider the BFS tree of $G$ with root $s$, where a node $u\in G$ which is $i$ distance from $s$ is located at the level $L_i$. As we know the length of the shortest path from $s$ to $t$ is $D$. This implies $t$ is some node at the level $L_D$. Now according to the previous strategy for trees, there are $k$ possible nodes in $G$ where the oracle can place $k$ pebbles, each at least $\frac{D}{k+1}$ {\it dist} apart from each other with respect to the level. The objective of the oracle is to place the pebbles odd {\it dist} apart from each other. Hence we describe the pebble placement strategy in the following manner:
\begin{itemize}
    \item If $\lceil\frac{D}{k+1}\rceil$ (=$l$) is even, then place the $j$-th (where $1\leq j \leq k$) pebble at the node $v_{\alpha+1}$ at level $L_{\alpha+1}$ along the shortest path $P$, where $\alpha={\lceil\frac{jD}{k+1}\rceil}$. The $j$-th pebble is placed $\lceil\frac{D}{k+1}\rceil+1$ {\it dist} apart from $(j-1)$-th pebble.
    \item Otherwise ($\lceil\frac{D}{k+1}\rceil$ is odd) the $j$-th pebble (where $1\leq j \leq k$) is placed at the node $v_{\alpha}$ at level $L_{\alpha}$ along the shortest path $P$. In this case $j$-th pebble is placed $\lceil\frac{D}{k+1}\rceil$ {\it dist} apart from the $(j-1)$-th pebble.
\end{itemize}
Below is a detailed description of the \textsc{TreasureHuntForBipartiteGraph} algorithm to find the treasure by an agent.
\begin{enumerate}

    \item \label{step-1Bipartite}Starting from $s$ the agent sets \textit{SearchNode}=$s$ and performs a breadth first search in lexicographically increasing order of outgoing port numbers until a pebble or treasure is found.
    \item \label{step-2Bipartite} If the treasure is found, the algorithm terminates.
    \item \label{step-3Bipartite} If the treasure is not found and a pebble is found at the node $v_l$ in $L_{l}$-th level, the agent performs two tasks. 
    \begin{enumerate}
    \item At first it stores the length of the path $l$  from {\it SearchNode} to the node containing the pebble. 
    
    \item Then it completes its search at the $L_{l-1}$-th level in a breadth first search technique. Whenever there is another encounter with the pebble, the agent stores the incoming port number of the incoming edge of that node containing the pebble.
    \end{enumerate}
    \item   Set \textit{SearchNode}=$v_l$.
    \item \label{saved} From {\it SearchNode}, the agent again performs a breadth first search of length $l$, ignoring the saved incoming port numbers, until the treasure or pebble is encountered (As per the pebble placement strategy, the distance between each pebble is $l$).
    \begin{enumerate}
        \item If the treasure is found, go to step \ref{step-2Bipartite}.
        \item If the treasure is not found and a pebble is found at the node $v_{2l}$ at level $L_{2l}$, then go to step \ref{step-3Bipartite}.
    \end{enumerate}
 
\end{enumerate}


Now in the next two lemmas, we first prove the correctness of our algorithm and further provide the complexity analysis.
\begin{lemma} \label{Correctness-Bipartite}
Given $k<\frac{D}{2}$ pebbles, the agent following \textsc{TreasureHuntForBipartiteGraph} algorithm successfully finds the treasure in a bipartite graph.
\end{lemma}

\begin{proof}
To prove the correctness of our algorithm, we need to ensure all the facts stated in lemma \ref{Correctness-GeneralGraphmultiMarker}, are resolved in the case of a bipartite graph. We deal with these issues in the following manner: 
\begin{itemize}
    \item The pebbles are placed odd distance (=$l$) apart. We know that bipartite graphs cannot have odd cycles (see proof in \cite{Lecture}). So, the agent searching from {\it SearchNode} a path of length $l$ cannot return to itself.
    
    \item Let, at some iteration of the algorithm, current {\it SearchNode}=$v_i$ at level $L_i$ of the BFS tree. Then the pebble encountered while searching all possible paths of length $l$ should be at level $L_j$, where $j>i$. This is because, the length of the path from the last {\it SearchNode} to $v_i$ is at least $l+1$. The reason being, the agent does not travel through the edges adjacent to $v_i$ whose ports are saved (refer step \ref{saved} of \textsc{TreasureHuntForBipartiteGraph} algorithm). This implies that the new {\it SearchNode} is always at the $L_j$-th level, where $j>i$. So, this means that the agent cannot go back to the already visited pebbles.
    
    \item The oracle places the pebbles $l$ distance apart. So, the agent while searching from {\it SearchNode}, all possible paths of length $l$, must encounter one pebble.
    
    \item The node containing the last pebble is chosen in such a manner by the oracle, such that the {\it dist} between treasure and that node is at most $l$. So, the agent while searching from {\it SearchNode} all paths of length $l$ must find the treasure.
\end{itemize}
All the above points guarantee that the agent successfully finds the treasure.
\qed
\end{proof}
\begin{theorem}
Given $k<\frac{D}{2}$, the agent finds the treasure in $\mathcal{O}(k\Delta^{\frac{D}{k}})$ time in a bipartite graph.
\end{theorem}

\begin{proof}
The agent starting from $s$ searches all possible paths of length $l$ (refer step \ref{step-1Bipartite} of \textsc{TreasureHuntForBipartiteGraph} algorithm). Now as the oracle places the pebbles $\lceil \frac{D}{k+1}\rceil$ (or $\lceil \frac{D}{k+1}\rceil+1$) distance apart (refer the pebble placement strategy), this implies the length of the path $l$ is either $\lceil \frac{D}{k+1}\rceil$ or $(\lceil \frac{D}{k+1}\rceil+1)$. So, the time required to search all possible paths of length $l$ in worst case is asymptotically $\mathcal{O}(\Delta^{\frac{D}{k+1}})$. Now from each pebble encountered, the agent searches all possible paths of length $l$ and then reaches the next pebble along the path $P$. As the agent successfully reaches the treasure by lemma \ref{Correctness-Bipartite}. Hence, this process goes on until the treasure is encountered (refer step \ref{saved} of \textsc{TreasureHuntForBipartiteGraph} algorithm). Now each such search of paths of length $l$ from $k$ many pebbles takes $\mathcal{O}(\Delta^{\frac{D}{k+1}})$ time. So, the total taken to find treasure in a bipartite graph is $\mathcal{O}(k\Delta^{\frac{D}{k+1}})=\mathcal{O}(k\Delta^{\frac{D}{k}})$.
\qed
\end{proof}

\section{Treasure Hunt Algorithm when $k\geq D$}\label{k>D}

In this section, we explore the case when $k=cD$ pebbles are provided by the oracle, where $c$ is any positive integer. We propose an algorithm that finds the treasure in $\mathcal{O}\left[cD{(\frac{\Delta}{2^{{c}/{2}}})}^2 + cD\right]$ time.

Let $G$ be a graph with maximum degree $\Delta \geq 10 (c+1)+6$. \footnote{The reason behind the limit of $\Delta$ is explained in \textit{Remark} \ref{note}}. Let  $\beta= 10 (c+1)+6$. The case where $\Delta < \beta$ is dealt with a different strategy, and it is explained ahead. The path $P$ from $s$ to $t$ may have two scenarios: \\
\textit{Scenario-1:} $P$ may not contain any node of degree $\beta$ which is similar to solving the case in which $G$ has maximum degree $\Delta < \beta$. \\
\textit{Scenario-2: } $P$ contains at least one node of degree $\beta$. 

All these cases are dealt separately and are discussed ahead. In similar structure, before proceeding to general graphs we first provide a description of our algorithm and pebble placement strategy in trees and then further extend our idea for general graphs keeping in mind the additional difficulties faced in case of general graph.

\subsection{Idea of Treasure Hunt in Tree for $k=cD$ Pebbles}\label{tree}Let $G$ be a rooted tree, where initial node $s$ is the root of the tree. The nodes at the level $L_i$ are located at a distance $i$ from $s$. The treasure $t$ is located at a distance $D$ from $s$ at the level $L_D$. Let $P=v_0,v_1,\cdots ,v_{D}$ (where $v_0$=$s$ and $v_{D}=t$) be the shortest path from $s$ to $t$.


If $k=D$ pebbles are given, then the oracle places a pebble on each $D$ many nodes along $P$, i.e., one pebble is placed on each $v_i$, where $0\le i \le D-1$ and $v_i\in P$. Now if more than $D$ pebbles are provided, i.e., $k>D$. Then along with placing $D$ pebbles on each node $v_i$, the oracle further places the remaining pebbles along the children of $v_i$'s. These remaining pebbles help the agent to reduce its search domain to find the next node $v_{i+1}$ along $P$  (where $0\le i \le D-1$). The agent from $v_i$, obtains a binary string by visiting the neighbors of $v_i$ along which some of the remaining $k-D$ pebbles are placed. This binary string gives the knowledge to the agent, of the collection of outgoing ports along which the agent must search in order to encounter the pebble placed at the node $v_{i+1}$. In the following part, we discuss how a string is represented with respect to the pebbles placed.

 \noindent\textbf{String Representation with Pebbles:} Among the neighbors of $v_i$ which are used for encoding a string: if the node contains a pebble it is termed as `1' in the $j$-th bit of the binary string, whereas no pebble represents `0'. Now there may be a scenario, where all the neighbors are not used for encoding. In order to learn where the encoding has ended, the following strategy is used \cite{gorain2021pebble}. Instead of a simple binary representation we provide a \textit{transformed }binary representation in which we replace `1' by `11' and `0' by `10'. This transformation ensures there is no `00' substring in the transformed binary string. As an example $\gamma=0010$ will be transformed to $\gamma^t = 10101110$. Hence when the agent finds two consecutive `0's it learns that the encoding has ended. The whole process is explained with the help of the following example.

\noindent\textbf{Example:} Given $k=2D$ pebbles, the example as shown in Fig. \ref{Fig-Tree2Dpebble} explains the execution the algorithm when the agent reaches the node $v_i\in P$ along the incoming port $\rho_0$. Let $deg(v_i)=12$ and the node $v_{i+1}$ connected to $v_i$ with the edge having outgoing port number $\rho_i$, where $\rho_i\in \{\rho_1,\cdots,\rho_{\frac{deg(v_i)}{2}-1}\}$ (i.e., $\rho_4$ to be exact). Then the pebbles for encoding are placed along the nodes corresponding to the outgoing ports  $\{\rho_{\frac{deg(v_i)}{2}},\cdots,\rho_{deg(v_i)-1}\}$. The $j$-th bit of the binary string is `1' if the node corresponding to the outgoing port number $\rho_{(\frac{deg(v_i)}{2}+(j-1))}$ contains a pebble, otherwise, if no pebble is found then the $j$-th bit is `0'. So, the agent currently at $v_i$ obtains the transformed binary string $\gamma^t=11$ (as the `00' obtained stops the agent from further search) by searching the nodes corresponding to the outgoing ports  $\{\rho_{\frac{deg(v_i)}{2}},\cdots,\rho_{deg(v_i)-1}\}$. Hence the binary string is $\gamma=1$. Now as the length of $\gamma$ obtained is 1, it divides the neighbor set of $\frac{deg(v_i)}{2}$ in to $2^{|\gamma|}$ (where $|\gamma|=1$) partitions each of size at most $\lceil\frac{deg(v_i)}{2^{1+1}}\rceil = \lceil \frac{12}{4} \rceil =3$. Further, it searches the outgoing ports corresponding to the 2nd partition (as $0$ represents the 1st partition, whereas $1$ represents the second partition of $\frac{deg(v_i)}{2}$ neighbors of $v_i$) out of $2^1$ partitions each consisting exactly 3 ports. This means the agent searches only the nodes corresponding to the outgoing ports $\rho_4$, $\rho_5$ and $\rho_6$ and finds the node $v_{i+1}$ containing a pebble via the outgoing port $\rho_4$.

\begin{figure}[h]
\centering
\includegraphics[width=0.8\textwidth]{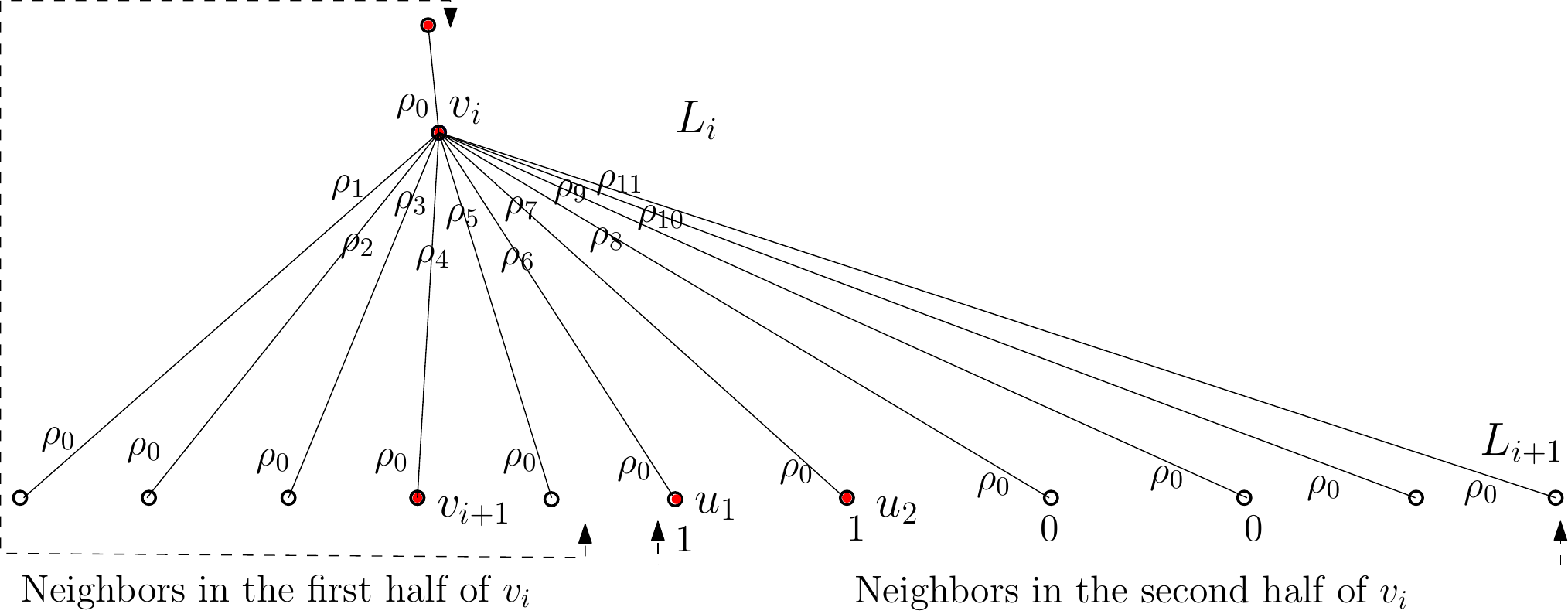}
\caption{Represents the encoding in order to help the agent reach $v_{i+1}$ from $v_i$. Pebbles are placed at the nodes $u_1$ and $u_2$ of the tree, to represent the {\it transformed} binary string $11$, which the agent obtains to search the nodes along a collection of outgoing port numbers for a pebble. This string localises the search along the nodes corresponding to the outgoing ports $\rho_3,\rho_4$ and $\rho_5$. Finally obtains a pebble at the node $v_{i+1}$ corresponding to the outgoing port $\rho_4$.}
\label{Fig-Tree2Dpebble}
\end{figure}


Now, this idea is simple for trees, but this exact idea will not work for any arbitrary graph. So, we make necessary modifications and explain them in the following section.

\subsection{Extending the idea for General Graphs} The above idea for trees cannot be directly extended to general graphs. It is because, any tree can be transformed into a rooted tree, the root being $s$. In which, the edges go from level $L_i$ to $L_{i+1}$ (where $i\ge 0$), creating an acyclic structure. The reason being, in a tree there is a unique path between two nodes, i.e., no two nodes have common children. Similarly, we can create a BFS tree of any arbitrary graph rooted at $s$. But any arbitrary graph may contain cycles. So, there may be edges going in between levels in the BFS tree as well. Now recalling the pebble placement idea for trees. The encoding in the neighbors of a node $v$ does not affect the encoding in the neighbors of node $u$ as there are no common children. But this is not true for general graphs. The encoding done for the node $u$ can hamper the encoding for the node $v$. To resolve this issue, we place the pebbles for encoding on high degree nodes that are not `close'. We call these high degree nodes as fat nodes which are defined below. A node is fat if its degree is at least $\beta$, where $\beta=10(c+1)+6$. Otherwise, the node is termed as \textit{light}. 

 Now we have the following cases and we deal with them separately:
\begin{itemize}
    \item \textit{Case-1:} Every node $v_i\in P$, $0\le i \le D-1$, is light.
    \item \textit{Case-2:} There exists at least one node in $P$ which is fat.
\end{itemize}
 \noindent\textit{Case-1:} In this case no encoding is needed. The oracle places a single pebble at each level of the BFS tree along the path $P$. So, the agent starting from $s$, sets \textit{SearchNode}$=s$. If a pebble is found at $s$ then it searches the neighbors having a pebble along the outgoing port $\{\rho_0,\cdots,\rho_{\frac{deg(s)}{2}}\}$.  Otherwise if no pebble is found at $s$ then it searches the neighbors having a pebble along the outgoing ports $\{\rho_{\frac{deg(s)}{2}+1},\cdots,\rho_{deg(s)-1}\}$. Whenever the next pebble is found at a node $v_1$, it sets the \textit{SearchNode}$=v_1$. At each subsequent steps the agent visits all the neighbors of the current \textit{SearchNode} for a pebble, except the incoming port which connects the current \textit{SearchNode} to the previous \textit{SearchNode} (i.e., except the port $\rho_0$ for the node $v_i$ in Fig. \ref{Fig-Tree2Dpebble}). This process will continue until the treasure is found. Now, since all the nodes along $P$ are light, hence their degree is less than $\beta$. So, the time needed to find the treasure is at most $\beta D$, where $\beta=10(c+1)+6$.
 
 \noindent\textit{Case-2:}  In this case encoding is needed as all the nodes along the path $P$ are not light. The encoding is done on the children of a set of nodes termed as {\it milestone}. The presence of each {\it milestone} helps the agent to localise the search domain for the next few nodes from along the path $P$. To define the first {\it milestone} node, we have the following cases in the BFS tree corresponding to $G$:
 \begin{itemize}
     \item \textit{Case-A:} The node $s$ is fat. This implies the first {\it milestone} is $s$.
     \item \textit{Case-B:} The node $s$ is light but the node $v_1$ at level $L_1$ is fat. This implies the first {\it milestone} is $v_1$.
     \item \textit{Case-C:} The nodes $s,v_1,\cdots,v_j$ (where $j\ge 2$) are light whereas the node $v_{j+1}$ at level $L_{j+1}$ is fat. This implies the first {\it milestone} is $v_{j+1}$.
 \end{itemize}
The subsequent milestones are defined recursively as follows. For $i\geq 1$, let the $i$-th milestone is in level $L_j$ (where $j \ge 0$). Then the $(i+1)$-th milestone node should be at level $L_k$, where $k-j\geq 5$, i.e., the distance between any two milestone is at least 5. This distance is maintained to avoid having a common neighbor between any two pair of milestones. Refer to the remark \ref{note} for detailed explanation. Since the agent has no knowledge about the underlying topology hence it cannot distinguish between light or milestone nodes. The placement of pebbles for encoding not only gives the binary representation but also determines whether a node is a milestone node or a light node (refer \textsc{CheckerForMilestone} algorithm). The pebble placement strategy is discussed in the following section.

\subsubsection{Pebble Placement:}\label{pebbleplacement}

There are two reasons for pebble placement. One is for giving the direction to the treasure along $P$. The other is for encoding which reduces the search domain for the next node along $P$. Pebbles are placed at every node along the path $P$, except at a node which is $2$ {\it dist} apart from a milestone node. More precisely, if $v_i$ is a milestone node at level $L_i$ along $P$, then no pebble is placed at the node $v_{i+2}$ at level $L_{i+2}$ along $P$ (refer the node $v_2$ in Fig. \ref{Fig-StateDiagram}, where $s$ is the first milestone). The goal of the agent is to find the next node of the path $P$ where the pebble is located. In the worst case, the agent may have to search all the neighbors. To reduce this search domain, encoding is incorporated. So, encoding will be done only at the neighbors of the fat node. Now the question is, which neighbors of the milestone are used for encoding. As the oracle knows which neighbor of the milestone along $P$ the pebble is placed, accordingly it uses the other half of the neighbors to place the pebbles for encoding. As shown in the Fig. \ref{Fig-StateDiagram}, where $s$ is a milestone and $v_1$ is the next node along $P$ and the pebbles for encoding are placed along the other half of neighbors of $s$.

\noindent\textit{Strategy for Encoding:} The number of available neighbors for encoding is $\frac{deg(v)}{2}$, where considering $v$ is a milestone. Out of $cD$ pebbles, $D$ number of pebbles are placed along the path $P$. The remaining $(c-1)D$ pebbles are used for encoding. The length of each encoding should be at most $c-1$. To distinguish between two separate encodings, the oracle leaves two consecutive neighbors without pebble. To understand the termination of the encoding another three consecutive neighbors are kept empty. To encode $\alpha$ many binary string, we need at least $\alpha((c-1)+2)+3$ neighbors. So, the relation between $deg(v)$ and $\alpha$ is $\frac{deg(v)}{2}\ge \alpha(c+1)+3$. Further, we define the set $\mathcal{R}$ as the set of two outgoing port numbers and all the incoming port numbers of the neighbors of the milestone along which pebbles are placed  corresponding to the nodes containing pebbles from $s$. The cardinality of $\mathcal{R}$ is at most $\alpha(c-1)+3$. The reason being, from a milestone $\alpha$ binary strings are encoded, each having length at most $(c-1)$ and a single pebble is placed along the desired path. As shown in Fig. \ref{Fig-Tree2Dpebble} the set $\mathcal{R}=\{\rho_4, \rho_6,\rho_0,\rho_0,\rho_0\}$, corresponding to the nodes $v_{i+1},u_1$ and $u_2$, respectively.

 \begin{figure}[h]
 \centering
 \includegraphics[width=0.8\textwidth]{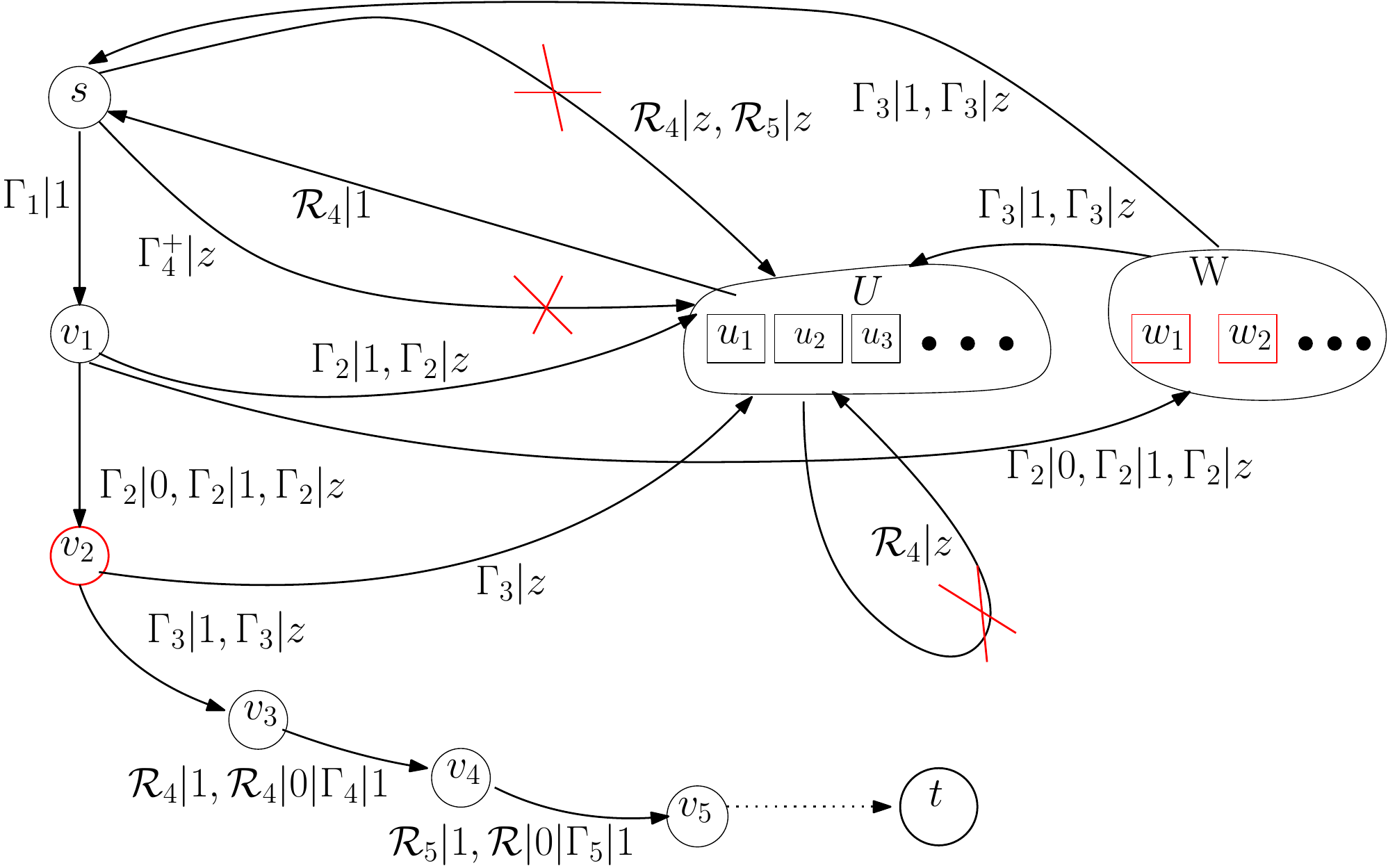}
 \caption{Directed graph representing the possible path traversal of agent from a Fat Node}
 \label{Fig-StateDiagram}
 \end{figure}

\begin{figure}[h]
 \centering
 \includegraphics[width=0.9\textwidth]{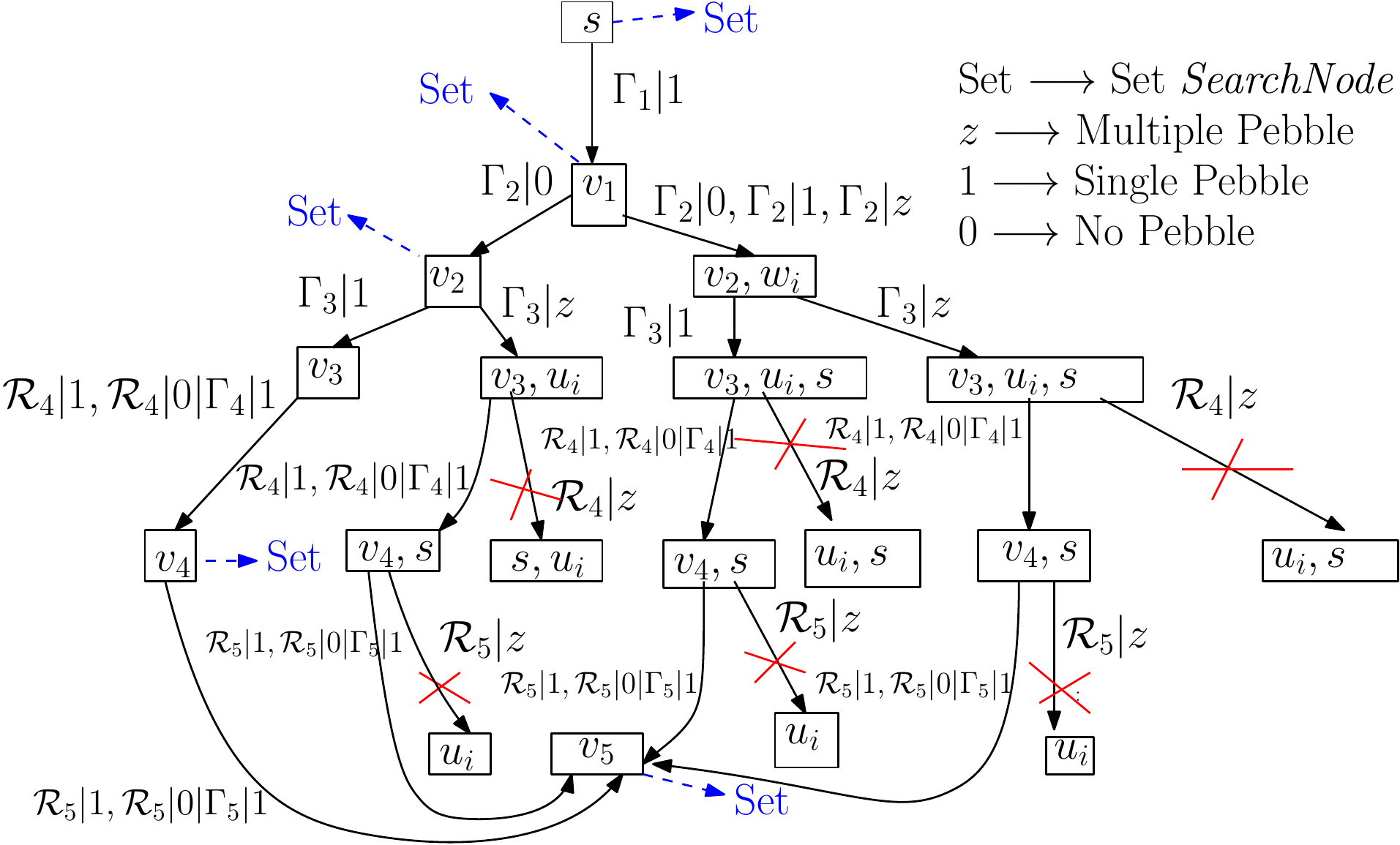}
 \caption{Flow Chart representing all possible traversals from a fat node $s$}
 \label{Fig-FlowChart}
 \end{figure}

Below is a detailed description of the \textsc{CheckerForMilestone} algorithm, for the agent to determine whether a node is a milestone or a marker.
\begin{enumerate}
    \item The agent currently at a node $v$, starts checking the node along the outgoing port $\rho_1$.
    \item If a pebble is found, then it searches the next four consecutive neighbors, i.e., the nodes with outgoing ports $\rho_2,\rho_3,\rho_4$ and $\rho_5$, respectively.
    \begin{enumerate}
        \item  If more than one pebble is found, then the agent concludes $v$ is a milestone and the encoding is done along the first half of its neighbors, except the parent. It further concludes that the next node along $P$ is present in the other half of its neighbors.
        \item Otherwise, if a single pebble is found, corresponding to the nodes with outgoing ports $\rho_1,\rho_2,\rho_3,\rho_4$ and $\rho_5$, then the node having pebble is the next node along $P$. It further checks the node with outgoing port $\rho_{\frac{deg(v)}{2}+1}$, if a pebble is found, then conclude that $v$ is a milestone, otherwise $v$ is light.
    \end{enumerate}
    \item Moreover if no pebble is found at the node with outgoing port $\rho_1$, then the agent checks the nodes with outgoing port  $\rho_{\frac{deg(v)}{2}},\rho_{\frac{deg(v)}{2}+1},\rho_{\frac{deg(v)}{2}+2}, \rho_{\frac{deg(v)}{2}+3}$ and $\rho_{\frac{deg(v)}{2}+4}$, respectively.
     \begin{enumerate}
         \item If more than one pebble is found, then the agent concludes $v$ is a milestone and the encoding is done along the second half of its neighbors, except the parent. It further concludes that the next node along $P$ is present in the other half of its neighbors.
         \item Otherwise, if a single pebble is found, corresponding to the nodes with outgoing ports $\rho_{\frac{deg(v)}{2}},\rho_{\frac{deg(v)}{2}+1},\rho_{\frac{deg(v)}{2}+2}, \rho_{\frac{deg(v)}{2}+3}$ and $\rho_{\frac{deg(v)}{2}+4}$, then the node having pebble is the next node along $P$. Moreover, conclude that the node $v$ is light.
     \end{enumerate}
    
\end{enumerate}

As shown in Fig. \ref{Fig-StateDiagram}, we create a directed graph representation consisting of all the possible paths that the agent can travel from a fat node $s$ towards the treasure $t$. The pebbles for encoding are placed along the neighbors of $s$. The set $U$ is a collection of nodes $\{u_1,\cdots \}$ which represent the nodes where pebbles are placed for encoding. The set $W$ is the collection of nodes $\{w_1,w_2,\cdots \}$, which are at the same level as $v_1$ and no pebbles are placed on them. The nodes along the desired path from $s$ to $t$ are depicted by circles, in which pebbles are placed at every node, except at $v_2$ (marked by a red circle) and $t$. Let $\Gamma_i$ be the integer value of the binary string $\gamma_i$ (where $1\le i \le \alpha$) encoded along the neighbors of $s$. The edge $(u,v)$ denoted by $\Gamma|n$ implies that the agent after searching some $\Gamma$ partition of $u$'s neighbor, encountered $n$ many pebbles. The edge $(u,v)$ denoted by $\mathcal{R}|n$ implies that the agent after searching the nodes corresponding to the set $\mathcal{R}$ of $u$'s neighbor, encountered $n$ many pebbles.The notation $\mathcal{R}|0|\Gamma$ along an edge $(u,v)$, represents the fact that the agent after searching the nodes corresponding to the set $\mathcal{R}$ of $u$'s neighbor and encounters no pebble, further it searches its $\Gamma$ partition of neighbors and encounters $n$ many pebbles. The red cross on an edge denotes that it is a path in which the agent detects inconsistency and stops further exploration along this path. The variable $z$ in Fig. \ref{Fig-StateDiagram}, represents any integer value greater than 1. 

Moreover, the Fig. \ref{Fig-FlowChart} is a generalized version of Fig. \ref{Fig-StateDiagram}. In this figure, all possible locations where the agent can travel by using the strings $\gamma_1,\cdots,\gamma_5$ are explained.

Further, any search from a node defined in \textsc{TreasureHuntForGraph} algorithm, is performed keeping in mind the fact that the agent does not search along its parent port. As shown in Fig. \ref{Fig-Tree2Dpebble} by the $\rho_0$ from $v_i$.

Below is a detailed description of the \textsc{TreasureHuntForGraph} algorithm for the agent to find the treasure.
\begin{enumerate}
    \item \label{step-1cD} The agent starting from $s$, sets {\it SearchNode}=$s$ checks for a pebble at $s$. If no pebble is found at $s$, then it searches the first half of its neighbors, for a node with a pebble. If a pebble is found at $s$, then it performs the \textsc{CheckerForMilestone} algorithm to check, whether the node is light or a milestone.
    \item \label{light} If $s$ is light, then it searches the second half of its neighbors until a treasure or a pebble is encountered. If the treasure is found then the algorithm terminates. Otherwise if a pebble is found at a node $v_1$, then set {\it SearchNode}=$v_1$.
    
    \item \label{step-3cD} If $s$ is a milestone, then it decodes the $\alpha$ many binary strings by visiting the second half of its neighbors (node $u_i$'s in Fig. \ref{Fig-Tree2Dpebble}). It stores the set $\mathcal{R}$ and performs the following task. 
    \begin{enumerate}
        \item \label{step-3acD}The agent first obtains the binary strings $\gamma_1,\cdots,\gamma_{\alpha}$ from the {\it transformed} binary strings. The length of each $\gamma_j$ ($1\le j \le \alpha$) is at most $\frac{c-1}{2}$, as length of the transformed binary string is at most $c-1$ (ref. to pebble placement of section \ref{pebbleplacement}).
        \item \label{step-3bcD}Further from $s$, i.e., {\it SearchNode}. The agent divides the neighbors of $s$ into $2^{|\gamma_1|}$ partitions. Each partition consisting of $\frac{deg(s)}{2^{|\gamma_1|}}$ neighbors. Then it searches $\Gamma_1$-th partition, where $\Gamma_1$ is the integer value of $\gamma_1$. If the treasure is found then the algorithm terminates. Otherwise a pebble is found at a node $v_1$ (say), set {\it SearchNode}= $v_1$ (as shown by the edge $(s,v_1)$ in Fig. \ref{Fig-StateDiagram} and in Fig. \ref{Fig-FlowChart}). 

        \item \label{milestone-c}From $v_1$, i.e., {\it SearchNode}. The agent searches $\Gamma_2$-th partition of neighbors of $v_1$ out of $2^{|\gamma_2|}$ partitions each consisting of at most $\frac{deg(v_1)}{2^{|\gamma_2|}}$ neighbors. There may be no pebble or at least one pebble found (refer the edges $(v_1,v_2)$, $(v_1,u_i)$ and $(v_1,w_i)$ in Fig. \ref{Fig-StateDiagram}).\\
        If at least one pebble encountered, it means that the agent  encountered pebbles placed at the nodes $u_i\in U$ which are at the same level as $v_1$, and are meant for encoding (refer the edges $(v_1,u_i)$, where $u_i\in U$ in Fig. \ref{Fig-StateDiagram} and the edges denoted by $\Gamma_2|1$ in Fig. \ref{Fig-FlowChart}). The desired path from $v_1$ is towards $v_2$ which is at the next level of $v_2$, which has no pebble (refer the edge $(v_1,v_2)$ with red circle in Fig. \ref{Fig-StateDiagram}). Further there are also nodes $w_i\in W$ at the same as $v_1$, which do not contain any pebbles as well (refer the edge $(v_1,w_i)$ in Fig. \ref{Fig-StateDiagram}). Now, the agent is unable to determine among them which one is the `shortest' path towards the treasure. To determine this fact, the agent on visiting the nodes without pebbles one at a time performs the following task. From each of these nodes, the agent uses the string $\gamma_3$ to search their respective $\Gamma_3$ partition of its neighbors and encounters at least one pebble. From each of these nodes with pebbles, the agent further searches the nodes corresponding to the set $\mathcal{R}$. Now in this process, the agent may face some inconsistencies. These inconsistencies will in turn help the agent reject the wrong paths, i.e., along the pebbles with nodes (refer to the nodes $w_i$ in Fig. \ref{Fig-StateDiagram}) at the same level as $v_1$ and ultimately guide it towards the desired path.
        
        The detailed process is explained as follows.
        
       So, irrespective of the number of nodes without pebbles is encountered from $v_1$ after searching $\Gamma_2$ partition of its neighbors, the agent visits each node without a pebble, one at a time by maintaining a stack. Then it searches its $\Gamma_3$ partition of its neighbors. If no pebble is encountered, then the agent returns to its parent. Otherwise, there can be a single or multiple pebbles encountered (refer to the edges with notation $\Gamma_3|1$ and $\Gamma_3|z$, respectively in Fig. \ref{Fig-StateDiagram}).\\
        \textit{If a single pebble is found}, then there are multiple possibilities, as shown by the edges denoted by $\Gamma_3|1$ in the Fig. \ref{Fig-FlowChart}. 
        \begin{description}
            \item {\it P1:} The agent currently at some node $w_i\in W$, encounters a pebble at a node in the previous level (refer the edge $(w_i,s)$ with notation $\Gamma_3|1$ in Fig. \ref{Fig-StateDiagram}).
            \item {\it P2:} The agent currently at some node $w_i\in W$, encounters a pebble at a node in the same level (refer the edge $(w_i,u_i)$ with notation $\Gamma_3|1$ in Fig. \ref{Fig-StateDiagram}).
            \item {\it P3:} The agent at the node $v_2$, encounters a pebble at a node in the next level of $v_2$, i.e., at $v_3$ (refer the edge $(v_2,v_3)$ denoted by $\Gamma^3|1$ as shown in Fig. \ref{Fig-StateDiagram}) which is indeed the desired path towards the treasure. 
        \end{description}
        \textit{If multiple pebbles are found}, then we have further possibilities, as shown by the edges denoted by $\Gamma_3|z$ in Fig. \ref{Fig-FlowChart}.
        \begin{description}
            \item {\it P1:} The agent currently at some node $w_i\in W$, encounters a pebble at a node in the previous level, i.e., at $s$ (refer the edge $(w_i,s)$ with notation $\Gamma_3|z$ in Fig. \ref{Fig-StateDiagram}) and all the remaining pebbles along the nodes in the same level, i.e., along $u_i$ (refer the edge $(w_i,u_i)$ with notation $\Gamma_3|z$ in Fig. \ref{Fig-StateDiagram}).
            \item {\it P2:} The agent currently at some node $w_i\in W$, encounters all the pebbles at a node in the same level, i.e., along $u_i$ (refer the edge $(w_i,u_i)$ with notation $\Gamma_3|z$ in Fig. \ref{Fig-StateDiagram}).
            
            \item {\it P3:} The agent currently at $v_2$, encounters a pebble at a node in the next level, i.e., at $v_3$ (refer the edge $(v_2,v_3)$ as shown in Fig. \ref{Fig-StateDiagram}) which is indeed the desired path towards the treasure. The remaining along the nodes in the previous level, as shown by the edge $(v_2,u_i)$ with notation $\Gamma_3|z$ in Fig. \ref{Fig-StateDiagram}.
        \end{description} 
        
        So, irrespective of the number of pebbles encountered, the agent visits each one of them and searches the nodes corresponding to the ports in the set $\mathcal{R}$. 
        {\it If no pebble is encountered}, then the agent is at $v_3$. In this case, the agent further searches the $\Gamma_4$ partition of its neighbors and encounters $v_4$ (refer the edge $(v_3,v_4)$ with notation $\mathcal{R}_4|0|\Gamma_4|1$ in Fig. \ref{Fig-StateDiagram}). From $v_4$, it further searches $\Gamma_5$ partition of its neighbors and finds $v_5$. It sets {\it SearchNode}=$v_5$. \\
        {\it If a single pebble is found} then we have the following possibilities.
        \begin{description}

            \item {\it P1:} If the agent is currently at some node in $u_i\in U$, then the pebble encountered is at the node $s$ (refer the edge $(u_i,s)$ in Fig. \ref{Fig-StateDiagram}). 
            
            \item {\it P2:} If the agent is at a node $v_3$, then the pebble encountered is at the node $v_4$ (refer the edge $(v_3,v_4)$ with notation $\mathcal{R}_4|1$ in Fig. \ref{Fig-StateDiagram}) which is the desired path. 
            \end{description}
            In this case, the agent searches the nodes corresponding to the ports in the set $\mathcal{R}$, from the node where a single pebble is encountered. Then we have further possibilities: 
            \begin{description}
                \item {\it P1:} If no pebble is encountered, then the agent is at $v_4$. In this case, the agent further searches the $\Gamma_5$ partition of its neighbors and encounters $v_5$ (refer the edge $(v_4,v_5)$ with notation $\mathcal{R}_5|0|\Gamma_5|1$ in Fig. \ref{Fig-StateDiagram}). It sets {\it SearchNode}=$v_5$.
                \item {\it P2:} If a single pebble is encountered (refer the edge $(v_4,v_5)$ with notation $\mathcal{R}_5|1$ in Fig. \ref{Fig-StateDiagram}), then this is the correct path, and the agent will reach to the node $v_5$, and set {\it SearchNode}=$v_5$.
                \item {\it P3:} If multiple pebbles are encountered, return to its parent (refer to all the crossed edges denoted by $\mathcal{R}_5|z$ in Fig. \ref{Fig-StateDiagram} and in Fig. \ref{Fig-FlowChart}).
            \end{description}
             
        {\it If multiple pebbles are found} along this search, then the agent returns to its parent, as referred by the crossed red edges denoted by $\mathcal{R}_4|z$  in Fig. \ref{Fig-StateDiagram} and in Fig. \ref{Fig-FlowChart}.
         
        In each case, by rejecting every wrong path (refered as crossed red edges in Fig. \ref{Fig-FlowChart}), the agent will ultimately return to the node $v_5$ (refer all the edges denoted as $\mathcal{R}_5|1$ and $\mathcal{R}_5|0|\Gamma_5|1$ in Fig. \ref{Fig-FlowChart}) and set {\it SearchNode}=$v_5$.

        \item \label{step-3dcD} Further from $v_5$, i.e., {\it SearchNode}. The agent searches the $\Gamma_6$ partition of $v_5$ and encounters a pebble at the node $v_6$. Then it sets {\it SearchNode}=$v_6$. This process continues until {\it SearchNode}=$v_{\alpha}$
    \end{enumerate}
    \item \label{step-4cD}If {\it SearchNode} is light, search all its neighbor until a pebble or treasure is encountered. If the treasure is found, then the algorithm terminates. If a pebble is found at a node $v_j$, set {\it SearchNode}=$v_j$, where ($j\ge 2$).
    \item \label{step-5cD}If {\it SearchNode} is a milestone, then it searches its corresponding half of neighbors determined by the algorithm \textsc{CheckerForMilestone}, and then go to step $2$.
   
\end{enumerate}


\begin{remark}\label{note}
 In this remark, the idea behind the value of $\beta=10(c+1)+6$ is discussed. See the step \ref{milestone-c} of the \textsc{TreasureHuntForGraph} algorithm. Observe that at least $5$ binary strings are required to check that the agent is moving along the desired path $P$ from a milestone (refer the path $s\longrightarrow v_1 \longrightarrow v_2 \longrightarrow v_3 \longrightarrow v_4 \longrightarrow v_5$ in Fig. \ref{Fig-StateDiagram}). Otherwise the agent may circle inside a loop and never reach the treasure (refer the path $s\longrightarrow v_1 \longrightarrow w_i \longrightarrow s$ in Fig. \ref{Fig-StateDiagram}). Hence the value of $\alpha$ $\geq 5$. As discussed in {\it Pebble placement} strategy, to encode $\alpha$ many binary strings, we need at least $5 (c+1)+3$ neighbors in one half of the neighbors of the milestone. As each milestone is a fat node, hence total the degree of a fat node must be at least $2(5 (c+1)+3)=10(c+1)+6$.

\end{remark}

\begin{lemma}
Given $k=c D$ pebbles, the agent following the \textsc{TreasureHuntForGraph} algorithm successfully finds the treasure.
\end{lemma}

\begin{proof}
To prove the correctness of the algorithm, we first ensure that the agent successfully moves from one milestone to another milestone along $P$. Secondly, we ensure that the agent moves towards the treasure along the path $P$, from a light node. Let us consider the {\it SearchNode} is $s$. We deal with the issues in the following manner:

\noindent \textit{ SearchNode is a milestone:} Observe that there are two possible nodes along  $P$ between two milestones, which may create a problem, i.e., these nodes may deviate the agent towards a wrong direction (refer to the nodes $v_1$ and $v_2$ in Fig. \ref{Fig-StateDiagram}). The reason being, these are the only nodes between two milestone along $P$ those may have more than one neighbors containing pebbles (refer the edges $(v_1,u_i)$ and $(v_2,u_i)$ in Fig. \ref{Fig-StateDiagram}). Now the strategy discussed in step \ref{milestone-c} of \textsc{TreasureHuntForGraph} algorithm, ensures that the agent always reaches the node $v_5$ along the path $P$ using the advice from 5 binary strings $\gamma_1,\cdots,\gamma_5$. This in turn implies that the agent successfully overcomes the two problematic nodes and moves along the desired direction. Hence, the agent by following the above algorithm successfully moves from one milestone node to another.
    
\noindent\textit{ SearchNode is light:} The agent searches all the neighbors of the {\it SearchNode} and finds a single pebble. It moves to that node containing the pebble and updates the {\it SearchNode} (refer step-\ref{light} of \textsc{TreasureHuntForGraph} algorithm).

Hence the above two explanations guarantee that the agent moves forward following the algorithm and successfully reaches the node where the treasure is residing.
\qed
\end{proof}

\begin{lemma}\label{lemma-intermediate}
The agent following \textsc{TreasureHuntForGraph} algorithm takes $\mathcal{O}\left({c(\frac{\Delta}{2^{\frac{c}{2}}})}^2+c\right)$ time to reach from a milestone to another milestone.
\end{lemma}
\begin{proof}
Consider $s$ as the first milestone on $P$ and the second milestone is a node $v_{\alpha}$ along $P$ which is located at $\alpha$ distance from $s$ (where $\alpha\ge 5$) along $P$. The time taken by the agent to reach $v_{\alpha}$ from $s$ is explained in the following cases:
\begin{itemize}
    \item {\it $s$ to $v_1$}: The agent starting from $s$, finds it a milestone, then obtains the $\alpha$ many binary strings (refer step \ref{step-3acD} in \textsc{TreasureHuntForGraph} algorithm), by visiting at most $\alpha(c+1)+3$ neighbors (refer pebble placement strategy \ref{pebbleplacement}) in $(\alpha(c+1)+3)$ time. It further searches its $\Gamma_1$ partition consisting of at most $\frac{\Delta}{2^{\frac{c-1}{2}}}$ neighbors and finds $v_1$ (refer step \ref{step-3bcD} in \textsc{TreasureHuntForGraph} algorithm) in $\mathcal{O}(\frac{\Delta}{2^{\frac{c-1}{2}}})$ time. So, the total time taken to reach $v_1$ is $\mathcal{O}(\frac{\Delta}{2^{\frac{c-1}{2}}}+\alpha(c+1)+3)=\mathcal{O}(\frac{\Delta}{2^{\frac{c}{2}}}+c)$.
    \item {\it $v_1$ to $v_5$:} The agent from $v_1$, searches its $\Gamma_2$ partition of neighbors in $\mathcal{O}(\frac{\Delta}{2^{\frac{c}{2}}})$ time. Now as no pebble is placed on $v_2$ along $P$ (refer pebble placement \ref{pebbleplacement}), the agent reaches each node without pebble and searches its corresponding $\Gamma_3$ partition of neighbors (refer step \ref{milestone-c} of \textsc{TreasureHuntForGraph} algorithm) in $\mathcal{O}({(\frac{\Delta}{2^{\frac{c}{2}}})}^2)$ time. Then from each pebble encountered it searches the nodes corresponding to the set $\mathcal{R}$ (refer step \ref{milestone-c} of \textsc{TreasureHuntForGraph} algorithm). The time required to perform this search is $\mathcal{O}({c(\frac{\Delta}{2^{\frac{c}{2}}})}^2)$. Now, the agent after cancelling all the inconsistent paths (refer step \ref{milestone-c} of \textsc{TreasureHuntForGraph} algorithm) may either reach $s$ or the node $v_4$.
    Further the agent searches the nodes corresponding to the set $\mathcal{R}$ from the encountered pebble. If a single pebble is found, then that node corresponding to the pebble encountered is $v_5$. If no pebble is found, then it searches  $\Gamma_5$ partition along its neighbor to finally reach $v_5$ in $\mathcal{O}(\frac{\Delta}{2^{\frac{c}{2}}}+c)$ time. So, the total time taken by the agent to reach $v_5$ from $v_1$ is $\mathcal{O}(c{(\frac{\Delta}{2^{\frac{c}{2}}})}^2+c))$.
    \item {\it $v_5$ to $v_{\alpha}$:} The agent after reaching $v_5$, searches its $\Gamma_6$ partition of neighbors and finds $v_6$ in $\mathcal{O}(\frac{\Delta}{2^{\frac{c}{2}}})$ time. The agent performs similar searches and ultimately reaches $v_{\alpha}$. The total time taken is $\mathcal{O}(\frac{\Delta}{2^{\frac{c}{2}}})$. 
\end{itemize}
Hence, the total time taken to reach $v_{\alpha}$ from $s$ is $\mathcal{O}\left(c{(\frac{\Delta}{2^{\frac{c}{2}}})}^2+c\right)$.
\qed
\end{proof}

\begin{theorem}
Given $k=c D$ pebbles, the agent following the \textsc{TreasureHuntForGraph} algorithm finds the treasure in $\mathcal{O}\left[cD{(\frac{\Delta}{2^{{c}/{2}}})}^2 + cD\right]$ time.
\end{theorem}
\begin{proof}
By lemma \ref{lemma-intermediate}, the agent takes $\mathcal{O}\left(c{(\frac{\Delta}{2^{\frac{c}{2}}})}^2+c\right)$ time to reach from one milestone to another, placed $\alpha$ distance apart. In the worst case, the agent may find a milestone placed successively at a gap of $\alpha$ distance apart, from the last milestone along the path $P$. Hence, the total time taken by the agent to reach the treasure is $\mathcal{O}\left[cD{(\frac{\Delta}{2^{{c}/{2}}})}^2 + cD\right]$.

\qed
\end{proof}

\section{Lower Bound}\label{LowerBound}

In this section, we provide a lower bound result on time of treasure hunt for the case when the number of pebbles $k$ is at most $D-1$.

Let $T$ be a complete tree with $n$ nodes and of height $D$ where the degree of the root $r$ and each internal node is $\Delta$. There are $\Delta \cdot (\Delta-1)^{D-1}$ leaves in $T$. Let $p= \Delta \cdot (\Delta-1)^{D-1}$ and  $u_1, \ldots, u_p$ be the leaves of $T$ in the lexicographical ordering of the shortest path from the root $r$ to the leaves. For $1 \le i\le p $, we construct an input $B_i$ as follows. The tree $T$ is taken as the input graph, $r$ as the starting point of the agent, and $u_i$ as the position of the treasure.  Let $\cal B$ be the set of all inputs $B_i$, $1 \le i\le p$. Let $\cal A$ be any deterministic treasure hunt algorithm executed by the mobile agent and let $\cal L$ be any pebble placement algorithm for the set of instances $\cal B$.  We prove the following theorem.

\begin{theorem}\label{th:lowerbound}
There exists a tree  with maximum degree $\Delta$ and diameter $D$ such that any deterministic algorithm must require $\Omega((\frac{k}{e})^{\frac{k}{k+1}}(\Delta-1)^{\frac{D}{k+1}})$-time for the treasure hunt using at most $k$ pebbles placed on the nodes of $T$.
\end{theorem}

\begin{proof} We define a set $W$ that contains distinct elements of $T$ such that $W= \{v \in T:~ \cal L~\text{places a pebble on}$ $v$ $\text{for some instance}~B_i \in $\cal B$\}$. Let $x$ be the cardinality of $W$, where $x\leq n$. The number of ways one can put $k$ pebbles in $x$ nodes is $x\choose k$. So the number of different possible placements of pebbles can be at most $x\choose k$. As we have $p$ many instances in  $\cal B$, by Pigeonhole principle, there exists at least one way of pebble placement which is used to find the treasure in at least $\frac{p}{\binom{x}{k}}$ many input instances from the set $\cal B$. Let $\mathcal{T}$ be the time taken by $\cal A$ to find the treasure. As the agent may need to look for the treasure on $\frac{p}{\binom{x}{k}}$ many different leaf nodes for some placements of pebbles, we can say $\mathcal{T}\geq \frac{p}{\binom{x}{k}}$.
On the other hand, $\mathcal{T}\geq x$ must hold for some instance in $\cal B$. From these two inequalities, we have, 
$\mathcal{T}\geq max\{x, \frac{p}{\binom{x}{k}}\}$. So, to find the least value of $x$ and hence $\mathcal{T}$, we need to solve the equation $x=\frac{p}{\binom{x}{k}}$. Putting the value of $p$, we get, $x\binom{x}{k}= \Delta \cdot (\Delta-1)^{D-1}$. Using the inequality $(\frac{xe}{k})^k\geq \binom{x}{k}$, we have, $\frac{x^{k+1}e^k}{k^k}\geq (\Delta-1)^D$. Finally, we get, $x\geq (\frac{k}{e})^{\frac{k}{k+1}}(\Delta-1)^{\frac{D}{k+1}}$. Hence the theorem. \hfill $\qed$
\end{proof}

\section{Conclusion}\label{conclude}

In this paper we study trade-off between number of pebbles $k$ and time for treasure hunt for $k = cD$, where $c\ge 1$. For $k <D$, our propose upper bound and lower bound on time of treasure hunt are close. For $k = cD$, we propose an algorithm for treasure hunt. Therefore, proving a tight lower bound result for both of the above cases is natural problem to solve in future. On the other hand, as the previous result \cite{gorain2021pebble} proves that the fastest possible treasure hunt algorithm can be achieved with $O(D \log \Delta)$ pebbles, it will be interesting to investigate the case when $k \in w(D)$ and $k \in o(D \log \Delta)$. 

We propose algorithms which have close upper and lower bound, when $k<D$. In future we will like to provide a more tighter lower bound. Further, when $k\ge D$, we have given only upper bound. A possible future work will be to propose a lower bound for this proof.
\bibliographystyle{splncs04}
\bibliography{bibliog}
\end{document}